\documentclass[11pt,a4paper]{article}

\usepackage{etoolbox}

\newtoggle{DEBUG}
\togglefalse{DEBUG}

\usepackage{todonotes}
\usepackage{adjustbox}
\usepackage{amsmath,mathrsfs,amsbsy,bm,amssymb,enumerate}

\usepackage[tikz]{bclogo}

\usepackage[latin1]{inputenc}
\usepackage[german,dutch,english]{babel}

\usepackage{amsthm}

\usepackage{graphicx,epsfig,color}
\usepackage{pgf,tikz}
\usetikzlibrary{arrows}

\usepackage{anysize}

\newcommand{\markit}{}

\usepackage{subfigure}

\providecommand{\wt}[1]{\widetilde{#1}}

\providecommand{\cone}{\mathrm{cone}}

\newcommand {\Z}	  {\mathbb{Z}}
\newcommand {\Q}	  {\mathbb{Q}}
\newcommand {\R}	  {\mathbb{R}}

\newcommand {\Area}{{\mathrm{Area}}}

\newcommand{\vol}{\mathrm{vol}}

\def\R{\mathbb R}

\def\poly{\mbox poly}

\renewcommand{\epsilon}{\varepsilon}

\renewcommand{\leq}{\leqslant}
\renewcommand{\geq}{\geqslant}

\theoremstyle{plain}
\newtheorem{theorem}{Theorem}
\newtheorem{lemma}[theorem]{Lemma}

\theoremstyle{definition}
\newtheorem{problem}{Problem}
\newtheorem{algorithm}{Algorithm}
\theoremstyle{remark}
\newtheorem{remark}{Remark}
 
\begin{document}

\title{Geometric Random Edge}

\author{Friedrich Eisenbrand\thanks{Email: friedrich.eisenbrand@epfl.ch}\\
EPFL
\and
Santosh Vempala\thanks{Email: vempala@gatech.edu}\\
Georgia Tech\\
}
\maketitle

\begin{abstract}
  \noindent 
  We show that a variant of the random-edge pivoting rule results in a
  strongly polynomial time simplex algorithm for linear programs
  $\max\{c^Tx \colon x \in \R^n, \, Ax\leq b\}$,
  whose constraint matrix $A$
  satisfies a geometric property introduced by Brunsch and R\"oglin:
  The sine of the angle of a row of $A$
  to a hyperplane spanned by $n-1$
  other rows of $A$
   is at least $\delta$.

  This property is a geometric generalization of $A$
  being integral and each 
  sub-determinant of $A$
  being bounded by $\Delta$
  in absolute value. In this case  $\delta \geq 1/(\Delta^2 n)$.
  In particular, linear programs defined by totally unimodular
  matrices are captured in this framework. Here $\delta \geq 1/ n$
  and Dyer and Frieze previously described a strongly polynomial-time
  randomized simplex algorithm for linear programs with $A$
  totally unimodular.
   
  The expected  number of pivots of the simplex algorithm is polynomial in the
  dimension and $1/\delta$ and independent of the number of
  constraints of the linear program. Our main result can be viewed as
  an algorithmic realization of the proof of small diameter for such
  polytopes by Bonifas et al., using the ideas of Dyer and Frieze.

\end{abstract}

\thispagestyle{empty}
\newpage

\setcounter{page}{1}

\section{Introduction} 
\label{sec:introduction}

Our goal is to solve a linear program 
\begin{equation}
  \label{eq:1}
  \begin{array}{c}
      \max c^Tx \\
      Ax \leq b
  \end{array}
\end{equation}
where $A \in \R^{m\times n}$ is of full column rank, $b \in \R^m$ and $c \in \R^n$. 
The rows of $A$ are denoted by $a_i, \, 1 \leq i \leq m$. We assume
without loss of generality  that the rows have Euclidean length one, i.e., that 
$\|a_i\|=1$ holds for $i=1,...,m$.  The rows shall have the following
\emph{$\delta$-distance property} (we use $\langle . \rangle$ to denote linear span):
\begin{quote} 
  For any $I \subseteq [m]$, and $j \in [m]$, if $a_j \notin \left<
    a_i \colon i \in I \right>$ then $d(a_j, \left< a_i \colon i \in I
  \right>) \geq \delta$. In other words, 
  if $a_j$ is not in the span of the $a_i,  i \in I $, then the distance of $a_j$
  to the subspace that is generated by the $a_i, \, i \in I$ is at
  least $\delta$.
\end{quote}
In this paper, we analyze the simplex algorithm~\cite{D143} to solve~\eqref{eq:1} with a
variant of the \emph{random edge} pivoting rule.  Our main result is a
strongly polynomial running time bound for linear programs satisfying
the $\delta$-distance property. 

{\markit
\begin{theorem}\label{thm:main} 
  There is a random edge pivot rule that solves a linear program using
  $\poly(n, 1/\delta)$ pivots in expectation.  The expected running
  time of this variant of the simplex algorithm is polynomial in $n,m$
  and $ 1/\delta$. 
\end{theorem}

The $\delta$-distance
property is a geometric generalization of the algebraic property of an integer matrix having small 
sub-determinants in absolute value. Recall that a $k\times k$
\emph{sub-determinant} of $A$
is the determinant of a sub-matrix that is induced by a choice of $k$
rows and $k$ columns of $A$. 
To see this, let  $A \in \Z^{m\times n}$
with each of its 
sub-determinants bounded by $\Delta$
in absolute value. Let $B \subseteq \{1,\dots,n\}$
be a \emph{basis of $A$},
i.e., an index set satisfying $|B| = n$
and $\langle a_i \colon i \in B \rangle = \R^n$.
Let $A_B$
be the sub-matrix of $A$
that is induced by the rows indexed by $B$
and let $w \in \Q^n$
be the column of $A_B^{-1}$
with $a_i^T w =1$. 
The distance of $a_i$
to $\langle a_j \colon j \in B - i \rangle$
is the absolute value of
\begin{equation}
  \label{eq:3}
  \frac{a_i^T w }{ \| w \|} 
\end{equation}
By Cramer's rule, each entry of $w$ is a $(n-1) \times (n-1)$ sub-determinant of $A_B$  divided by $\det(A_B)$. The division by $\det(A_B)$ cancels out in \eqref{eq:3}.  After this cancellation, 
the numerator is an integer and the denominator is at most $\sqrt{n} \cdot \Delta$.  This shows that the absolute value of~\eqref{eq:3} is at least $1/(\sqrt{n}  \cdot \Delta)$.  To bound $\delta$ from below, this distance  needs to be divided by $\|a_i \| \leq \sqrt{n} \cdot \Delta$ since the rows of the  matrix should be scaled to length one when we measure the distance.  Thus an integer matrix satisfies the $\delta$-distance property with $\delta \geq 1/(n \cdot \Delta^2)$.

\medskip 
This shows that our result is   an extension of a randomized simplex-type
algorithm of Dyer and Frieze~\cite{MR1274170} that solves linear
programs~(\ref{eq:1}) for \emph{totally unimodular} ($\Delta = 1$)  $A \in \{0,\pm
1\}^{m\times n}$ and arbitrary $b$ and $c$.  In this case,  $\delta \geq 1/n^2$.


\subsection*{Related work} 
\label{sec:related-work}

General linear programming problems can be solved in weakly polynomial
time~\cite{Khachiyan79,MR779900}. This means that the number of basic
arithmetic operations performed by the algorithm is bounded by the
binary encoding length of the input.  It is a longstanding open
problem whether there exists a \emph{strongly polynomial time}
algorithm for linear programming. Such an algorithm would run in
time polynomial in the dimension and the number of inequalities on a
RAM machine. For general linear
programming, the simplex method is
sub-exponential~\cite{MSW96,DBLP:conf/stoc/Kalai92}. 

The $\delta$-distance
property is a geometric property and not an algebraic one. In fact,
the matrix $A$
can be irrational even. This is contrast to a result of
Tardos~\cite{MR861043}, who gave a strongly polynomial time algorithm
for linear programs whose constraint matrix has integer entries that
are bounded by a constant.  The algorithm of Tardos is not a
simplex-type algorithm.

Spielman and Teng~\cite{MR2145860} have shown that the simplex
algorithm with the \emph{shadow-edge pivoting rule} runs in expected
polynomial time if the input is randomly perturbed. This
\emph{smoothed analysis} paradigm was subsequently applied to many
other algorithmic problems. Brunsch and
R\"oglin~\cite{brunsch2013finding} have shown that, given two vertices
of a linear program satisfying the $\delta$-distance
property, one can compute in expected polynomial time a path joining
these two vertices with $O(m n^2 / \delta^2)$
edges in expectation with the shadow-edge rule. However, the two
vertices need to be known in advance. The authors state the problem of
finding an optimal vertex w.r.t. a given objective function vector $c$
in polynomial time as an open problem. We solve this problem and
obtain a path whose length is independent on the number $m$
of inequalities.

Klee and Minty~\cite{MR0332165} have shown that the simplex method is
exponential if Dantzig's original pivoting rule is applied. More
recently, Friedmann, Hansen and Zwick~\cite{MR2931978} have shown that
the \emph{random edge} results in a superpolynomial number of pivoting
operations. Here, random edge means to choose \emph{an improving} edge
uniformly at random. The authors also show such a lower bound for
\emph{random facet}. Friedmann~\cite{friedmann2011subexponential} also
recently showed a superpolynomial lower bound for \emph{Zadeh's}
pivoting rule.  Nontrivial, but exponential upper bounds for random
edge are given in~\cite{gartner2007two}.

Bonifas et
al.~\cite{BDEHN12} have shown that the diameter of a polytope defined by
an integral  constraint matrix $A$ whose sub-determinants are bounded
by $\Delta$ is polynomial in $n$ and $\Delta$ and independent of the
number of facets. In the setting of the $\delta$-distance property,
their proof yields a polynomial bound in $1/\delta$ and the dimension
$n$ on the diameter that is independent of $m$. Our result is an extension of this result in the setting of linear programming. We show that there is a variant of the simplex algorithm that uses a number of pivots that is polynomial in  $1/\delta$ and the dimension
$n$.

\subsection*{Assumptions} 

Throughout we  assume that $c$ and the rows of $A$, denoted by  $a_i, \, 1\leq i \leq m$, have Euclidean norm $\|\cdot\|_2$ one. 
We also assume that
the linear program is \emph{non-degenerate}, meaning that for each feasible
point $x^*$, there are at most $n$ constraints that are satisfied by
$x^*$ with equality. It is well known that this assumption can be made without loss of generality~\cite{Schrijver86}.

\section{Identifying an element of the optimal basis}
\label{sec:getting-close-c}


Before we describe our variant of the random-edge simplex algorithm,
we explain the primary goal, which is to identify one inequality of
the optimal basis.  Then, we can continue with the search for other
basis elements by running the simplex algorithm in one dimension
lower.

Let $K = \{x \in \R^n \colon Ax \leq b \}$ be the \emph{polyhedron} of
\emph{feasible solutions} of~(\ref{eq:1}). Without loss of generality,
see Section~\ref{sec:finding-an-initial}, we can assume that $K$ is a
\emph{polytope}, i.e., that $K$ is bounded.  Let $v \in K$ be a
\emph{vertex}. 
The \emph{normal cone} $C_v$ of $v$ is the set of vectors $w \in \R^n$ such
that, if $c^Tx$ is replaced by $w^Tx$ in (\ref{eq:1}), then $v$ is an
optimal solution of that linear program. Equivalently, let $B
\subseteq \{1,\dots,m\}$ be the $n$ row-indices with $a_i^T v = b_i,
\, i \in B$, then the normal cone of $v$ is the set $\cone\{a_i \colon
i \in B\}=\{ \sum_{i \in B} \lambda_i a_i \colon\lambda_i \geq 0, \, i
\in B\}$.  
The cones $C_u$ and $C_v$ of two vertices $u\neq
v$ intersect if and only if $u$ and $v$ are neighboring vertices of
$K$. In this case, they intersect in the common facet 
$\cone\{ a_i \colon 
i \in B_u
\cap B_v\}$, where $B_u$ and $B_v$ are the indices of tight
inequalities of $u$ and $v$ respectively, see 
Figure~\ref{fig:5}.

\begin{figure}[h]
  \centering
  \iftoggle{DEBUG}{}
  {
    \begin{tikzpicture}[scale=0.5]
    \clip (-7 ,-5) rectangle (18,5);

    \coordinate (A) at (-4,-1);
    \coordinate  (B) at (-2,2);
    \coordinate  (C) at (3,2);
    \coordinate  (D) at (4,-1);
    \coordinate  (E) at (1,-3);

    \filldraw[line width = 1pt, fill=gray!40!white](A) -- (B) -- (C) -- (D) -- (E) -- (A);

    \draw[line width = 1pt, shift=(A)] (0,0) -- ( -1.11 , -2.78 ); 
    \draw[line width = 1pt, shift=(A)] (0,0) --     ( -2.49 , 1.66 );

    \draw[line width = 1pt, shift=(B)] (0,0) --      ( -2.49 , 1.66 );
    \draw[line width = 1pt, shift=(B)] (0,0) -- (0,3);
    
    \draw[line width = 1pt, shift=(C)] (0,0) -- (0,3);
    \draw[line width = 1pt, shift=(C)] (0,0) -- ( 2.84 , 0.94 );

    \draw[line width = 1pt, shift=(D)] (0,0) -- ( 2.84 , 0.94 );
    \draw[line width = 1pt, shift=(D)] (0,0) -- ( 1.66 , -2.49 );

    \draw[line width = 1pt, shift=(E)] (0,0) -- ( 1.66 , -2.49 );
    \draw[line width = 1pt, shift=(E)] (0,0) -- ( -1.11 , -2.78 );

    {
      \filldraw (B) [ blue!80,fill=blue!80]  circle (3pt);
      
      \node [anchor = north west] at (B) {$u$};
    }

      \fill[red,opacity=.5, ] (-2,2) -- ( -4.49 , 3.66 ) -- (-2,5) -- (-2,2);
      
v

      \filldraw (C) [ blue!80,fill=blue!80]  circle (3pt);
      
      \node [anchor = north east] at (C) {$v$};
      
      \fill[red,opacity=.5, ] (3,2) -- ( 5.84 , 2.94 ) -- (3,5) -- (3,2);

      \begin{scope}

      \coordinate []  (F) at (14,0);
      
      \clip (F) circle [radius = 4.2cm]; 


    \draw[line width = 1pt, shift=(F)] (0,0) --      ( - 4.98 , 3.32 );

    \draw[line width = 1pt, shift=(F)] (0,0) -- (0,6);

    \draw[line width = 1pt, shift=(F)] (0,0) -- ( 5.68 , 1.88 );

    \draw[line width = 1pt, shift=(F)] (0,0) -- ( 2.32 , -4.98 );
    \draw[line width = 1pt, shift=(F)] (0,0) -- ( -2.22 , -5.46 );

    \fill[shift=(F),red,opacity=.5] (0,0) --     ( -4.98 , 3.32 )-- (0,6);    
    \node[align=left] at (13,1.8) {$C_u$};

    \fill[shift=(F),red,opacity=.5] (0,0) -- (0,6) -- ( 5.68 , 1.88 );    
    \node[align=left] at (15,1.8) {$C_v$};

    \filldraw (F) [ black]  circle (3pt);
    \node[align=right] at (14.5,-0.2) {$0$};

    
    

  \end{scope}     
\end{tikzpicture}
}

  \caption{Two neighboring vertices $u$ and $v$ and their normal-cones $C_u$ and $C_v$. }
  \label{fig:5}
\end{figure}
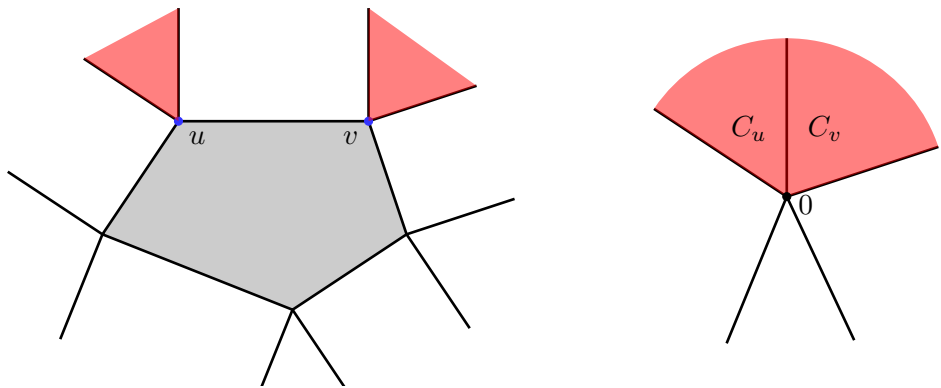


%
%

\medskip 
\noindent 
Suppose  that we found a point $c' \in \R^n$  together with a vertex $v$ of $K$ whose normal cone $C_v$  contains $c'$   such that: 
\begin{equation}
  \label{eq:4}
  \|c - c'\| < \delta/(2 \cdot n).  
\end{equation}
%
The following variant of a lemma proved by Cook et
al.~\cite{CookGerardsSchrijverTardos86} shows that we then can 
identify at least one index of the optimal basis of~(\ref{eq:1}). We provide a proof of the lemma in the appendix. 
\begin{lemma}
 \label{lem:4}
 Let $c' \in \R^n$ and 
 let $B \subseteq \{1,\dots,m\}$ be the optimal basis of the linear program~(\ref{eq:1}) and let $B'$ be an optimal basis of the linear program~(\ref{eq:1}) with $c$ being replaced by $c'$. Consider the conic combination 
 \begin{equation}
   \label{eq:2}
   c' = \sum_{j \in B'} \mu_j a_j. 
 \end{equation}
 For  $k \in B' \setminus B$, one has 
 \begin{displaymath}
   \| c - c' \| \geq \delta \cdot \mu_k.
 \end{displaymath}
\end{lemma}

Assuming~\eqref{eq:4} and following the notation of the lemma, let $B'$ be the optimal basis of
the linear program with objective function $c'x$. Since
$\|c\|=1$ we have $\|c'\| > 1 - \delta / (2\cdot n)$. Because the rows of $A$ have norm one, this implies   that
there exists a $k \in B'$ with $\mu_k > 1/n \cdot \left(1 - \delta /
  (2\cdot n) \right)$. Each such  $k $ must be in $B$ since $\delta \cdot
\mu_k > \delta/n \left(1 - \delta / (2\cdot n) \right) \geq
\delta/(2\cdot n)$.


\medskip 

Once we have identified an index $k $ of the optimal basis $B$ of
(\ref{eq:1}) we  set this inequality to equality and let
the simplex algorithm search for the next element of the basis on the
induced face of $K$. This is in fact a $n-1$-dimensional linear
program with the $\delta$-distance property. We explain now why this is the case. 

Suppose that the element from the optimal basis is $a_1$. Let $U \in
\R^{n \times n}$ be a non-singular orthonormal matrix that rotates $a_1$ into
the first unit vector, i.e.
  \begin{displaymath}
    a_1^T \cdot U = e_1^T. 
  \end{displaymath}
  The linear program $\max\{c^Tx \colon x \in \R^n, \, Ax\leq b\}$ is
  equivalent to the linear program $\max\{c^T U\cdot x \colon x \in \R^n, \, A
  \cdot U \cdot x \leq b\}$. Notice that this transformation preserves
  the $\delta$-distance property. Therefore we can assume that $a_1 $ is the first unit vector.  

  Now we can set the constraint $x_1 \leq b_1$ to equality and
  subtract this equation from the other constraints such that they do
  not involve the first variable anymore. The $a_i$ are in this way
  projected into the orthogonal complement of $e_1$. We scale the
  vectors and right-hand-sides with a scalar $\geq 1$ such that they
  lie on the unit sphere and are now left with a linear program with
  $n-1$ variables that satisfies  the $\delta$-distance property as we show now.  \footnote{A similar fact holds for totally unimodular constraint matrices, see, e.g.,\cite[Proposition 2.1, page~540]{NemhauserWolsey89} meaning that after one has identified an element of the optimal basis, one is left with a linear program in dimension $n-1$ with a totally unimodular constraint matrix. A similar fact fails to hold for integral matrices with sub-determinants bounded by $2$.}

\begin{lemma}
  \label{lem:1}
  Suppose that the vectors
  $a_1,...,a_m$ 
  satisfy the $\delta$-distance property, then $a_2^*,...,a_m^*$, after being  scaled to unit length, 
  satisfy the $\delta$-distance property as well, where $a_i^*$ is the
  projection of $a_i$ onto the orthogonal complement of $a_1$.
\end{lemma}

\begin{proof}
  Let $I \subseteq \left\{2,\dots,m\right\}$ and $j \in \left\{ 2,
    \dots, m\right\}$ such that $a_j^*$ is not in the span of the
  $a_i^*, \, i \in I$. Let $d(a^*_j, \left< a^*_i \colon i \in I
  \right>) = \gamma$. Clearly, $ d\left(a_j^*, \left< a_i \colon i \in I
    \cup \{1\} \right> \right) \leq \gamma$ and since $a_j^*$ stems from $a_j$
  by subtracting a suitable scalar multiple of $a_1$, we have $ d(a_j,
  \left< a_i \colon i \in I \cup \{1\} \right> \leq \gamma$ and
  consequently $\gamma \geq \delta$.
\end{proof}

\medskip   
\noindent 
This shows that we can solve the linear programming problem efficiently  if
 we there is an efficient algorithm that solves the following  problem.

  
\bigskip 
\begin{adjustbox}{minipage=\linewidth,fbox}
  \begin{problem}\label{prob:1}
    Given an initial vertex of the linear
    program~\eqref{eq:1}, compute a vector $c' \in \R^n$
    together with a vertex $v$ of~\eqref{eq:1} such that
    \begin{enumerate}[a)]
    \item The vector $c ' $ \label{item:3}
      is contained in  the normal cone $C_v$ of $v$, and  \label{item:1} 
    \item $\|c - c'\| < \delta/(2 \cdot n)$ hold. \label{item:2} 
    \end{enumerate}         
  \end{problem}
\end{adjustbox}  
\medskip  

\noindent 
Notice that condition \ref{item:3}) is equivalent to $v$ being an optimal solution of the linear program~\eqref{eq:1} wherein $c$ has been replaced by $c'$. 


\section{A random walk  controls the simplex  pivot operations} 
\label{sec:random-edge-pivoting}

Next we show how to solve Problem~\ref{prob:1} with the simplex algorithm wherein   the pivoting is carried out by a random walk. Consider the function 
\begin{equation}
  \label{eq:6} 
  f(x) = \exp( - \|x - \alpha \cdot c  \|_1 ),  
\end{equation}
where $\alpha \geq 1$ is a constant whose value will be determined later. 
Imagine that  we could sample a point $x^*$  from the distribution with density proportional to \eqref{eq:6} and that, at the same time, we are provided with a vertex $v$ that is an optimal solution of the linear program~\eqref{eq:1} where $c$ has been replaced by $c' = x^* / \alpha$. We are interested in the  probability that  $c' = x^* / \alpha$ together with $v$ is not a solution of Problem~\ref{prob:1}. This happens if $\|c' - c \|_2 \geq \delta / 2n$ and since $\|\cdot \|_1 \geq \|\cdot \|_2$ one has 
\begin{equation}
  \label{eq:7}
  \|x^* -\alpha \cdot c \|_1 \geq \alpha \cdot \delta / 2n  
\end{equation}
in this case. The probability of the event~\eqref{eq:7} is equal to the probability that a random $y^*$ chosen with a density proportional to $\exp(-\|y\|_1)$ has $\ell_1$ norm at least $\alpha \cdot \delta / 2n$. 
Since at least  one component of $y$ needs to have absolute value at least $\alpha \cdot \delta / 2n^2$ to satisfy~\eqref{eq:7}, this probability is upper bounded by 
\begin{equation}
  \label{eq:11}
n   \int_{\alpha \cdot \delta / 2n^2}^\infty  \exp(-x)\, dx  =   n / \exp(\alpha \cdot \delta / 2n^2) 
\end{equation}
by applying  the union bound. Thus if $\alpha \geq 2n^3 / \delta$, this probability is exponentially small in $n$. 

Approximate sampling for log-concave distributions can be dealt with by random-walk techniques~\cite{ApplegateK91}. 
As in the paper of Dyer and Frieze~\cite{MR1274170} we combine these techniques with the simplex algorithm to keep track of the optimal vertex of the current point of the random walk.

 Remember that the normal cones $C_v$ of the vertices $v$ partition the space $\R^n$. Each of these cones is now again partitioned into 
countably infinitely many parallelepipeds
whose axes are parallel to the unit vectors defining the cone of length $1/n^2$
see Figure~\ref{fig:6}.  
More precisely, a given cone $C_v = \{ \sum_{i \in B} \lambda_i a_i
\colon\lambda_i \geq 0, \, i \in B\}$ is partitioned into translates
of the parallelepiped
\begin{displaymath}  
  \left\{ \sum_{i \in B} \lambda_i \cdot a_i \colon 0 \le
    \lambda_i \le 1/n^2, \, i \in B \right\}.
\end{displaymath}
The \emph{volume} of such a
parallelepiped $P$ is $\vol(P) =(1/n^2)^n |\det(A_B)|$ where $A_B$ is the sub-matrix of $A$ consisting of the rows indexed by the basis $B$. 
For a parallelepiped $P$, we define
\[
f(P) = f(z_P)\vol(P) 
\]
where $z_P$ is the center of $P$. 
\begin{figure}[h]
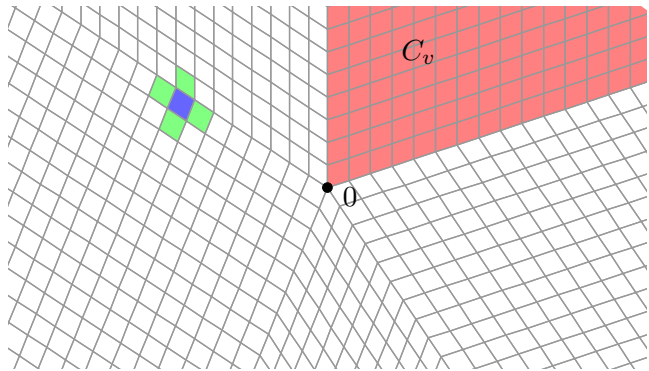

  \centering  
  \iftoggle{DEBUG}{
    }
    {
   \begin{tikzpicture}[scale=0.6]
   \coordinate []  (F) at (0,0);

   \clip  (-7,-4) rectangle  (7,4);

    \fill[red,opacity=.5] (0,0) -- (0,9) -- ( 8.52 , 2.82 );

    \input{One_Parallelepiped_and_Neighbors} 
     
     \input{parallelepipeds_picture_isoperimetric}

\node[align=left] at (2,3) {$C_v$};

    \filldraw (F) [ black]  circle (3pt);
    \node[align=right] at (0.5,-0.2) {$0$}; 
       
 \end{tikzpicture}
}
  \caption{The partitioning of $\R^2$ into parallelepipeds for the polytope in Figure~\ref{fig:5} together with the illustration of a parallelepiped (blue) and its neighbors (green).}
  \label{fig:6}
\end{figure}
The state space of the random walk will be all parallelepipeds used to
partition all cones, a countably infinite collection. One \emph{iteration of the random walk} is now as follows. 
\begin{quote}
\fbox{\parbox{.8\textwidth}{
  At the parallelepiped $P$, pick one of the $2n$
  parallelepipeds that share a facet with $P$
  uniformly at random, i.e. with probability $(1/2)\, n$.
  Let this parallelepiped be $P'$
  and go to it with probability $\frac{1}{2}\min\{1, f(P')/f(P)\}$.
} 
}
\end{quote}
\noindent 
Notice that $P'$ can be a parallelepiped that is contained in a different cone. If we then make the transition to $P'$ we perform a \emph{simplex pivot}. In this way we can always  keep track of the normal cone containing the 
current parallelepiped. 
We assume that
we are given some vertex of the polytope, $x^0$ and its associated
basis $B_0 \subseteq\{1,\dots,m\}$ to start with. We explain in the appendix how this assumption can be removed. We start the random walk at the parallelepiped of the cone $C_{x_0}$ that has vertex $0$. The algorithm is completely described below. 

\bigskip 
\begin{adjustbox}{minipage=.97\linewidth,fbox}
\begin{algorithm}[Geometric random edge]\label{alg:1}
~\\ ~ \medskip 

\begin{tabular}{l p{0.7\textwidth}}   
{\bf Input}: & An LP specified by $A,b,c$; a basic feasible solution $x^0$ of $Ax \le b$  and an associated basis $B_0$.  \\

{\bf Output}: & If successful, an element of the optimal basis  
\end{tabular}

\begin{enumerate}
\item Let $f(P)=\exp{(-\|z_P- \alpha c\|_1)} \,\vol(P)$. Start with the parallelepiped $P_0$ in $C_{x_0}$ containing the point $0$. 
\item Repeat for $\ell$ iterations:
\begin{itemize}
\item If $c$ is in the current cone, the algorithm stops. It has found the optimal basis. 
\item Otherwise pick a neighboring parallelepiped, say $P'$, uniformly at random.
\item Go to $P'$ with probability $\frac{1}{2}\min \left\{ 1, \frac{f(P')}{f(P)}\right\}$.\\ 
(this is a pivot whenever $P'$ and $P$ are in different cones.)
\end{itemize}
\item If the center $z_P$ of the final parallelepiped satisfies $\|z_P / \alpha - c\| < \delta/2n$, then { return} an element of the optimal basis as described in Section~\ref{sec:getting-close-c}. Return {\bf \tt failure} otherwise. 
\end{enumerate}

\end{algorithm}
\end{adjustbox}

\bigskip 

\noindent 
Before we analyze the convergence of the random walk we first 
explain why the parallelepipeds in the partition of a cone  are generated by the basis vectors scaled down to length $1/n^2$. 
In short, this is because the  value of the function  $f(x) = \exp(-\|x -\alpha \cdot c\|_1)$ does not vary too much for points in neighboring parallelepipeds. More precisely, we can state the following lemma, where $e = \exp(1)$ denotes the \emph{Euler constant}. 

\begin{lemma}
  \label{lem:2}
  Let $P_1$ and $P_2$ be two parallelepipeds that share a common facet and let $x_1, x_2 \in P_1 \cup P_2$. If $n \geq 4$  then 
  \begin{displaymath}
    f(x_1) / f(x_2) \leq  e. 
  \end{displaymath}
\end{lemma}
\begin{proof}
  The Euclidean distance between any two points in one of the two parallelepipeds is at most $1/n$ and thus $\|x_1 - x_2\| \leq 2/n$. This implies $\|x_1-x_2\|_1 \leq 2/\sqrt{n}\leq 1$ since $n \geq 4$. Consequently 
  \begin{eqnarray*}
    f(x_1) / f(x_2) & =  &\exp(- \|x_1 - \alpha\cdot c\|_1 + \|x_2 - \alpha\cdot c\|_1 ) \\
    & \leq & \exp( \|x_1 - x_2\|_1 ) \\
    & \leq & \exp(1). 
  \end{eqnarray*}
\end{proof}

\section{Analysis of the walk} 
\label{sec:analysis-walk}
 
We assume that the reader is familiar with basics of \emph{Markov chains} (see e.g., \cite{Lovasz1993}). The random walk above with its transition probabilities is in fact  a Markov chain $\mathcal{M} = (\mathscr{P}, p)$ with a countably infinite \emph{state space} $\mathscr{P}$  which is the set of parallelepipeds  and \emph{transition probabilities}  $p\colon \mathscr{P} \times \mathscr{P} \longrightarrow \R_{\geq 0}$. 
%
Let $Q\colon \mathscr{P} \longrightarrow \R_{\geq 0}$ be a probability distribution on $\mathscr{P}$. The distribution $Q$ is called \emph{stationary} if  for each $P \in \mathscr{P}$ one has 
\begin{displaymath}
  Q(P) = \sum_{P' \in \mathscr{P}} p(P',P) \cdot Q(P'). 
\end{displaymath}
Our  Markov chain is  {\em lazy} since, at each step, with probability $\geq 1/2$  
it does nothing.  The \emph{rejection sampling} step where we step from $P$ to $P'$ with probability $1/2 \cdot \min\{1, f(P')/f(P)\}$ is called a \emph{Metropolis filter}. Laziness and the Metropolis filter ensure that $\mathcal{M}= (\mathscr{P},p)$ has  the unique stationary distribution  (see e.g., Section 1, Thm 1.4 of \cite{lovasz1993random} or Thm 2.1 of \cite{VemSurvey}):
\begin{displaymath}
  Q(P) = f(P) / \sum_{P' \in \mathscr{P}} f(P') 
\end{displaymath}
which  is a discretization  of the continuous distribution with density $2^{-n} exp(\|x - \alpha c \|_1)$  from~\eqref{eq:6}.

\medskip 

Performing $\ell$ steps of the walk induces a distribution $Q^{\ell}$ on $\mathscr{P}$ where $Q^\ell(P)$ is the probability that the walk is in the parallelepiped $P$ in the end. In the limit, when $\ell$ tends to infinity,   $Q^\ell$ converges to $Q$. 
We now show that, due to the $\delta$-distance property of the matrix $A$, the walk quickly converges to $Q$. More precisely,  we only need to run it for a polynomial number (in $n$ and $1/\delta$) iterations. Then  $Q^\ell$  will be sufficiently close to $Q$ which shows that Algorithm~\ref{alg:1} solves  Problem~\ref{prob:1} with high probability.

To prove convergence of the walk, we bound the
conductance of the underlying Markov chain \cite{SJ89}. 
The conductance of $\mathcal{M}$ is 
{\markit
  \begin{equation}
    \label{eq:8}
    \phi = \min_{ \substack{S \subseteq \mathscr{P}  \colon\\ 0 <Q(S)\leq 1/2}}\frac{\displaystyle \sum_{P \in S, P' \in \mathscr{P}\setminus S} Q(P)  \cdot p(P,P')}{ Q(S)}. 
  \end{equation}
}%
Jerrum and Sinclair~\cite{SJ89}   related the conductance to the convergence of a finite
Markov chain to its stationary distribution.  Lov\'asz and
Simonovits~\cite{lovasz1993random} extended their result to general Markov
chains and in particular to Markov chains with a countably infinite set of states like our chain  $\mathcal{M} = (\mathscr{P},p)$. We state their theorem in our setting. 
\begin{theorem}[{\cite[Corollary~1.5]{lovasz1993random}}]\label{thm:LS}
  Let  $Q^{\ell}$ be the distribution obtained after $\ell$ steps of the Markov chain  $\mathcal{M} = (\mathscr{P},p)$ started at the initial parallelepiped $P_0$. Then for any $T \subseteq \mathscr{P}$ one has 
  \begin{displaymath}
    |Q^{\ell}(T)-Q(T)| \le Q(P_0)^{-1/2}  \left(1-\frac{\phi^2}{2}\right)^{\ell}.
  \end{displaymath}
\end{theorem}
\noindent
The rate of convergence is thus $O(1/\phi^2)$. Our goal is now to bound $\phi$ from below. 

\subsection*{Bounding the conductance}

Inspecting~\eqref{eq:8}, we first note that the transition probability $p(P,P')$ is zero, unless the parallelepiped $P'$ is a neighbor of $P$, i.e., shares a facet with $P$.  In this case one has 
\begin{displaymath}
  p(P,P') = \frac{1}{4n} \min\{1, f(P') / f(P)\}. 
\end{displaymath}
The  ratio $f(P') / f(P)$ can be bounded from below by  $\delta / e$. 
This is because $f(z_{P'})/f(z_P) \geq 1/e$ by Lemma~\ref{lem:2} and since  $\vol(P') / \vol(P) \geq \delta$. The latter inequality is a consequence of the $\delta$-distance property, as the ratio $\vol(P') / \vol(P)$ is equal to the ratio of the  height of $P'$ and the height of $P$  measured from the common facet that they share respectively. This ratio is at east $\delta$. Consequently we have for neighboring parallelepipeds $P$ and $P'$, 
\begin{displaymath}
   p(P,P') \geq \delta / (4 e n).
\end{displaymath}
 Clearly $Q(P) \cdot p(P,P') = Q(P') \cdot p(P',P)$
which implies that the conductance can be bounded by 
\begin{equation}
  \label{eq:5}
  \phi \geq  \min_{ \substack{S \subseteq \mathscr{P}  \colon\\ 0 <Q(S)\leq 1/2}}({\displaystyle  \delta   }/{4en})  \, \frac{Q(N(S))}{Q(S)}. 
\end{equation}
where $N(S) \subseteq \mathscr{P} \setminus S$ denotes the \emph{neighborhood} of $S$, which is the set of parallelepipeds $P' \notin S$ for which there exists a neighboring $P \in S$. 

\medskip 

We next  make use of the following isoperimetric inequality that was shown by Bobkov and Houdré~\cite{bobkov1997isoperimetric}  for more general product  probability measures. For our function $f(x) = \exp(-\|x - \alpha\cdot c\|_1)$  it reads as follows, where $\partial(A)$ denotes the boundary of $A$. 

\begin{theorem}[\cite{bobkov1997isoperimetric}]\label{thm:iso}
Let $f(x) = \exp(-\|x - \alpha \cdot c\|_1)$. For any  measurable set $A \subset \R^n$ with $0<2^{-n} \int_A f(x) \, dx <1/2$ one has 
\[
\int_{\partial A} f(x)\, dx \ge  \frac{1}{2 \sqrt{6}} \int_A f(x)\, dx. 
\]
\end{theorem}
Theorem~\ref{thm:iso} together with the $\delta$-distance property  yields a lower bound on ${Q(N(S))}/{Q(S)}$ as follows.  Each parallelepiped $P' \in N(S)$ has at least  facet  $F$  at the boundary $\partial(\cup_{P \in S} P)$, see Figure~\ref{fig:8}.
Lemma~\ref{lem:2} implies that $\int_F f(x) \, dx \leq e \cdot  \Area(F) \cdot  f(z_{P'})$ and since the height of $P'$ w.r.t. $F$ is at least $\delta$ one has $ \Area(F) \cdot \delta \leq \vol(P')  $ which implies 
\begin{displaymath}
  \int_F f(x) \, dx  \leq e/\delta \cdot f(z_{P'}) \vol(P')  = e / \delta \cdot   f(P'). 
\end{displaymath}
Since each $P' \in N(S)$ has $2n$ facets, this implies 
\begin{equation}
  \label{eq:9}
2n\,e / \delta  \cdot   \sum_{P' \in N(S)} f(P') \geq \frac{1}{2 \sqrt{6}} \displaystyle \int_{\cup_{P \in S} P} f(x)\, dx \geq \frac{1}{2\, e\, \sqrt{6}} \sum_{P \in S} f(P)
\end{equation}
where we used Lemma~\ref{lem:2} again in the last inequality. Thus 
$Q(N(S))/Q(S) \geq \delta / (4 e^2 \sqrt{6} n)$ which implies a bound of $\Omega(\delta^2/n^2)$  on the conductance. 

\begin{figure}[h]
  \centering  
  \iftoggle{DEBUG}{}
  {
 \begin{tikzpicture}[scale=0.5]
   \coordinate []  (F) at (0,0);
      
   \clip  (-7,-5.5) rectangle  (8,5.5);     
   
   { 

   }

   \begin{scope}
     \input{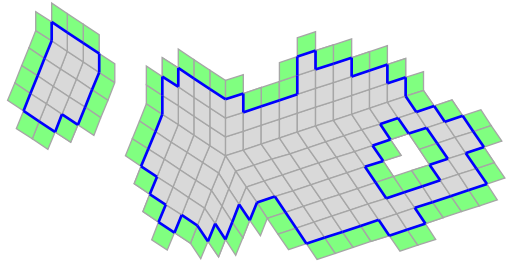}
   \end{scope}         
   
   

    \begin{scope}
       \draw[  blue, line width = 1pt ] ( 4.26 , 1.92 )
-- ( 4.74 , 2.08 )
;  \draw[  blue, line width = 1pt ] ( 4.74 , 2.08 )
-- ( 4.74 , 1.58 )
;  \draw[  blue, line width = 1pt ] ( 0.47 , 1.15 )
-- ( 0.94 , 1.31 )
;  \draw[  blue, line width = 1pt ] ( 0.94 , 1.31 )
-- ( 1.42 , 1.47 )
;  \draw[  blue, line width = 1pt ] ( 1.42 , 1.47 )
-- ( 1.89 , 1.63 )
;  \draw[  blue, line width = 1pt ] ( 3.32 , 2.10 )
-- ( 3.79 , 2.26 )
;  \draw[  blue, line width = 1pt ] ( 3.79 , 2.26 )
-- ( 4.26 , 2.42 )
;  \draw[  blue, line width = 1pt ] ( 4.26 , 2.42 )
-- ( 4.26 , 1.92 )
;  \draw[  blue, line width = 1pt ] ( 0 , 1.50 )
-- ( 0.47 , 1.65 )
;  \draw[  blue, line width = 1pt ] ( 0.47 , 1.65 )
-- ( 0.47 , 1.15 )
;  \draw[  blue, line width = 1pt ] ( 1.89 , 1.63 )
-- ( 1.89 , 2.13 )
( 1.89 , 1.63 )
;  \draw[  blue, line width = 1pt ] ( 2.37 , 2.29 )
-- ( 2.84 , 2.44 )
;  \draw[  blue, line width = 1pt ] ( 2.84 , 2.44 )
-- ( 3.32 , 2.60 )
;  \draw[  blue, line width = 1pt ] ( 3.32 , 2.60 )
-- ( 3.32 , 2.10 )
;  \draw[  blue, line width = 1pt ] ( 1.89 , 2.63 )
-- ( 2.37 , 2.79 )
;  \draw[  blue, line width = 1pt ] ( 2.37 , 2.79 )
-- ( 2.37 , 2.29 )
;  \draw[  blue, line width = 1pt ] ( 1.89 , 2.13 )
-- ( 1.89 , 2.63 )
( 1.89 , 2.13 )
;  \draw[  blue, line width = 1pt ] ( 1.30 , -1.08 )
-- ( 0.83 , -1.24 )
;  \draw[  blue, line width = 1pt ] ( 1.30 , -1.08 )
-- ( 1.58 , -1.50 )
( 1.30 , -1.08 )
;  \draw[  blue, line width = 1pt ] ( 1.58 , -1.50 )
-- ( 1.86 , -1.92 )
( 1.58 , -1.50 )
;  \draw[  blue, line width = 1pt ] ( 2.61 , -2.17 )
-- ( 2.13 , -2.33 )
;  \draw[  blue, line width = 1pt ] ( 1.86 , -1.92 )
-- ( 2.13 , -2.33 )
( 1.86 , -1.92 )
;  \draw[  blue, line width = 1pt ] ( 3.08 , -2.02 )
-- ( 2.61 , -2.17 )
;  \draw[  blue, line width = 1pt ] ( 3.56 , -1.86 )
-- ( 3.08 , -2.02 )
;  \draw[  blue, line width = 1pt ] ( 4.03 , -1.70 )
-- ( 3.56 , -1.86 )
;  \draw[  blue, line width = 1pt ] ( 3.67 , -0.29 )
-- ( 3.95 , -0.71 )
;  \draw[  blue, line width = 1pt ] ( 3.95 , -0.71 )
-- ( 4.23 , -1.13 )
;  \draw[  blue, line width = 1pt ] ( 4.78 , -1.96 )
-- ( 4.31 , -2.12 )
;  \draw[  blue, line width = 1pt ] ( 4.03 , -1.70 )
-- ( 4.31 , -2.12 )
( 4.03 , -1.70 )
;  \draw[  blue, line width = 1pt ] ( 3.87 , 0.27 )
-- ( 4.15 , -0.14 )
;  \draw[  blue, line width = 1pt ] ( 4.15 , -0.14 )
-- ( 3.67 , -0.29 )
;  \draw[  blue, line width = 1pt ] ( 4.23 , -1.13 )
-- ( 4.70 , -0.97 )
( 4.23 , -1.13 )
;  \draw[  blue, line width = 1pt ] ( 4.98 , -1.38 )
-- ( 5.26 , -1.80 )
;  \draw[  blue, line width = 1pt ] ( 5.26 , -1.80 )
-- ( 4.78 , -1.96 )
;  \draw[  blue, line width = 1pt ] ( 4.07 , 0.84 )
-- ( 4.34 , 0.43 )
;  \draw[  blue, line width = 1pt ] ( 4.34 , 0.43 )
-- ( 3.87 , 0.27 )
;  \draw[  blue, line width = 1pt ] ( 5.45 , -1.23 )
-- ( 4.98 , -1.38 )
;  \draw[  blue, line width = 1pt ] ( 4.70 , -0.97 )
-- ( 5.18 , -0.81 )
( 4.70 , -0.97 )
;  \draw[  blue, line width = 1pt ] ( 4.54 , 1.00 )
-- ( 4.07 , 0.84 )
;  \draw[  blue, line width = 1pt ] ( 5.93 , -1.07 )
-- ( 5.45 , -1.23 )
;  \draw[  blue, line width = 1pt ] ( 5.18 , -0.81 )
-- ( 5.65 , -0.65 )
( 5.18 , -0.81 )
;  \draw[  blue, line width = 1pt ] ( 4.74 , 1.58 )
-- ( 5.02 , 1.16 )
;  \draw[  blue, line width = 1pt ] ( 5.29 , 0.74 )
-- ( 4.82 , 0.59 )
;  \draw[  blue, line width = 1pt ] ( 4.54 , 1.00 )
-- ( 4.82 , 0.59 )
( 4.54 , 1.00 )
;  \draw[  blue, line width = 1pt ] ( 5.37 , -0.24 )
-- ( 5.65 , -0.65 )
( 5.37 , -0.24 )
;  \draw[  blue, line width = 1pt ] ( 5.37 , -0.24 )
-- ( 5.85 , -0.08 )
( 5.37 , -0.24 )
;  \draw[  blue, line width = 1pt ] ( 6.40 , -0.91 )
-- ( 5.93 , -1.07 )
;  \draw[  blue, line width = 1pt ] ( 5.49 , 1.32 )
-- ( 5.77 , 0.90 )
;  \draw[  blue, line width = 1pt ] ( 5.02 , 1.16 )
-- ( 5.49 , 1.32 )
( 5.02 , 1.16 )
;  \draw[  blue, line width = 1pt ] ( 5.29 , 0.74 )
-- ( 5.57 , 0.33 )
( 5.29 , 0.74 )
;  \draw[  blue, line width = 1pt ] ( 5.57 , 0.33 )
-- ( 5.85 , -0.08 )
( 5.57 , 0.33 )
;  \draw[  blue, line width = 1pt ] ( 6.32 , 0.07 )
-- ( 6.60 , -0.34 )
;  \draw[  blue, line width = 1pt ] ( 6.60 , -0.34 )
-- ( 6.88 , -0.75 )
;  \draw[  blue, line width = 1pt ] ( 6.88 , -0.75 )
-- ( 6.40 , -0.91 )
;  \draw[  blue, line width = 1pt ] ( 6.24 , 1.06 )
-- ( 6.52 , 0.64 )
;  \draw[  blue, line width = 1pt ] ( 5.77 , 0.90 )
-- ( 6.24 , 1.06 )
( 5.77 , 0.90 )
;  \draw[  blue, line width = 1pt ] ( 6.52 , 0.64 )
-- ( 6.80 , 0.23 )
;  \draw[  blue, line width = 1pt ] ( 6.80 , 0.23 )
-- ( 6.32 , 0.07 )
;  \draw[  blue, line width = 1pt ] ( -0.27 , -1.80 )
-- ( -0.46 , -2.27 )
;  \draw[  blue, line width = 1pt ] ( -0.46 , -2.27 )
-- ( -0.74 , -1.85 )
;  \draw[  blue, line width = 1pt ] ( 0.36 , -1.29 )
-- ( 0.18 , -1.76 )
;  \draw[  blue, line width = 1pt ] ( 0.18 , -1.76 )
-- ( -0.00 , -2.22 )
;  \draw[  blue, line width = 1pt ] ( -0.00 , -2.22 )
-- ( -0.27 , -1.80 )
;  \draw[  blue, line width = 1pt ] ( 0.83 , -1.24 )
-- ( 0.64 , -1.71 )
;  \draw[  blue, line width = 1pt ] ( 0.64 , -1.71 )
-- ( 0.36 , -1.29 )
;  \draw[  blue, line width = 1pt ] ( -1.84 , 0.64 )
-- ( -1.66 , 1.10 )
;  \draw[  blue, line width = 1pt ] ( -3.32 , 2.21 )
-- ( -3.51 , 1.75 )
( -3.32 , 2.21 )
;  \draw[  blue, line width = 1pt ] ( -4.76 , 2.58 )
-- ( -4.57 , 3.05 )
;  \draw[  blue, line width = 1pt ] ( -2.03 , 0.18 )
-- ( -1.84 , 0.64 )
;  \draw[  blue, line width = 1pt ] ( -3.51 , 1.75 )
-- ( -3.69 , 1.29 )
( -3.51 , 1.75 )
;  \draw[  blue, line width = 1pt ] ( -4.94 , 2.12 )
-- ( -4.76 , 2.58 )
;  \draw[  blue, line width = 1pt ] ( -1.80 , -0.56 )
-- ( -2.22 , -0.28 )
;  \draw[  blue, line width = 1pt ] ( -2.22 , -0.28 )
-- ( -2.03 , 0.18 )
;  \draw[  blue, line width = 1pt ] ( -3.88 , 0.82 )
-- ( -4.30 , 1.10 )
;  \draw[  blue, line width = 1pt ] ( -3.69 , 1.29 )
-- ( -3.88 , 0.82 )
( -3.69 , 1.29 )
;  \draw[  blue, line width = 1pt ] ( -5.13 , 1.65 )
-- ( -4.94 , 2.12 )
;  \draw[  blue, line width = 1pt ] ( -0.74 , -1.85 )
-- ( -1.15 , -1.57 )
;  \draw[  blue, line width = 1pt ] ( -1.57 , -1.30 )
-- ( -1.99 , -1.02 )
;  \draw[  blue, line width = 1pt ] ( -1.99 , -1.02 )
-- ( -1.80 , -0.56 )
;  \draw[  blue, line width = 1pt ] ( -4.48 , 0.63 )
-- ( -4.90 , 0.91 )
;  \draw[  blue, line width = 1pt ] ( -4.30 , 1.10 )
-- ( -4.48 , 0.63 )
( -4.30 , 1.10 )
;  \draw[  blue, line width = 1pt ] ( -4.90 , 0.91 )
-- ( -5.31 , 1.19 )
;  \draw[  blue, line width = 1pt ] ( -5.31 , 1.19 )
-- ( -5.13 , 1.65 )
;  \draw[  blue, line width = 1pt ] ( -1.34 , -2.04 )
-- ( -1.76 , -1.76 )
;  \draw[  blue, line width = 1pt ] ( -1.76 , -1.76 )
-- ( -1.57 , -1.30 )
;  \draw[  blue, line width = 1pt ] ( -1.15 , -1.57 )
-- ( -1.34 , -2.04 )
( -1.15 , -1.57 )
;  \draw[  blue, line width = 1pt ] ( -0.41 , 1.77 )
-- ( 0 , 1.50 )
;  \draw[  blue, line width = 1pt ] ( -0.83 , 2.05 )
-- ( -0.41 , 1.77 )
;  \draw[  blue, line width = 1pt ] ( -1.24 , 1.83 )
-- ( -1.24 , 2.33 )
;  \draw[  blue, line width = 1pt ] ( -1.24 , 2.33 )
-- ( -0.83 , 2.05 )
;  \draw[  blue, line width = 1pt ] ( -1.66 , 1.10 )
-- ( -1.66 , 1.60 )
;  \draw[  blue, line width = 1pt ] ( -1.66 , 1.60 )
-- ( -1.66 , 2.10 )
;  \draw[  blue, line width = 1pt ] ( -1.66 , 2.10 )
-- ( -1.24 , 1.83 )
;  \draw[  blue, line width = 1pt ] ( -3.74 , 2.99 )
-- ( -3.32 , 2.71 )
;  \draw[  blue, line width = 1pt ] ( -3.32 , 2.21 )
-- ( -3.32 , 2.71 )
( -3.32 , 2.21 )
;  \draw[  blue, line width = 1pt ] ( -4.16 , 3.27 )
-- ( -3.74 , 2.99 )
;  \draw[  blue, line width = 1pt ] ( -4.57 , 3.05 )
-- ( -4.57 , 3.55 )
;  \draw[  blue, line width = 1pt ] ( -4.57 , 3.55 )
-- ( -4.16 , 3.27 )
;
    \end{scope}

 \end{tikzpicture}
}

 \caption{An illustration of Theorem~\ref{thm:iso} in our setting. The gray parallelepipeds are the set $S$ and the boundary $\partial S$ is in blue. The green parallelepipeds are $N(S)$.}
  \label{fig:8}
\end{figure}

{\markit

\begin{lemma}\label{lem:conductance}
The conductance of the random walk on the parallelepipeds  is $\Omega(\delta^2 / n^2)$.
\end{lemma}

\subsection*{Bounding the failure probability of the algorithm}

Algorithm~\ref{alg:1} does not return an element of the optimal basis if 
\begin{displaymath}
  \|z_P - \alpha \cdot c\|_2 \geq \delta / 2n.
\end{displaymath}
What is the probability  of this event if $P$ is sampled according to $Q$, the stationary distribution of the walk? If this happens, then the final parallelepiped  is contained in the set
\begin{equation}
  \label{eq:10}
  \{x \in \R^n \colon  \|x - \alpha \cdot c\|_2 \geq \alpha \delta / 2n - 1/n\}
\end{equation}
since the diameter of a parallelepiped is at most $1/n$. Let $T \subseteq \mathscr{P}$ be the set of parallelepipeds that are contained in the set~\eqref{eq:10}. To estimate $Q(T)$ we use the notation $\wt{T}$ to denote the set of points in the union of the parallelepipeds, i.e., $\wt{T} = \cup_{P \in T}P$. Using Lemma~\ref{lem:2}, we see that 
\begin{displaymath}
  Q(T) = \sum_{P \in T} f(P) /  \sum_{P \in \mathscr {P}} f(P) \leq e^2 \cdot    2^{-n} \cdot \int_{\wt{T}}f(x) \, dx. 
\end{displaymath}
Clearly  $2^{-n} \cdot \int_{\wt{T}}f(x) \, dx$ is at most the probability that a random point  $x^* \in \R^n$ sampled from the distribution with density  $2^{-n}\cdot \exp(-\|x \|_1)$ has $\ell_2$-norm at least $ \alpha \cdot \delta / 2n - 1/n \geq \alpha \cdot  \delta / 4n$ where we assume that $\alpha \geq 2 n^3 / \delta$  in the last inequality. Arguing as in~\eqref{eq:11} we conclude that 
\begin{displaymath}
  Q(T) \leq e^2 \cdot n / \exp( \alpha \cdot \delta / 4n^2). 
\end{displaymath}
Thus is $\alpha = 4n^3 / \delta$, then 
\begin{equation}
  \label{eq:14}
  Q(T) \leq e^2\cdot  n / \exp(n) 
\end{equation}
which is exponentially small in $n$. 

\medskip

How many steps $\ell$ of the walk are necessary until 
\begin{equation}
  \label{eq:15}
  |Q^\ell(T) - Q(T)| \leq \exp(-n) 
\end{equation}
holds? Remember that the conductance $\phi \geq \xi \cdot  (\delta/n)^2$ for some constant $\xi >0$. Thus~\eqref{eq:15} holds if 
\begin{equation}
  \label{eq:16}
  \sqrt{e^2 2^{n} / f(P_0)} \left(1 - \xi\cdot  (\delta/n)^2\right)^\ell < \exp(-n). 
\end{equation}
Now $f(P_0) = f(z_{P_0}) \vol(P_0) \geq \exp (-\|2 \alpha c \|_1) n^{-n} \delta^n \geq \exp(-8 \cdot n^{3.5}/\delta ) (\delta/n)^n$. Thus~\eqref{eq:16} holds if
\begin{displaymath}
  \left(1 - \xi\cdot  (\delta/n)^2\right)^\ell  < \exp(-  \nu  n^{3.5}/\delta) (\delta / n)^{n},
\end{displaymath}
where $\nu>0$ is an appropriate constant. Using the inequality $(1+x) \leq exp(x)$ and taking logarithms on both sides, this holds if 
\begin{equation}
  \label{eq:18}
  - \ell \cdot  \xi\cdot  (\delta/n)^2 < -\nu  n^{3.5}/\delta + n \ln \delta / n
\end{equation}
and thus $\ell = \Omega(n^{5.5} / \delta^3)$ is a polynomial lower bound on $\ell$ such that \eqref{eq:15} holds. With the union bound we thus have our main result. 

\begin{theorem}
  \label{thr:1}
  If Algorithm~\ref{alg:1} performs $\Omega(n^{5.5} / \delta^3)$
  steps of the random walk, then the probability of failure is bounded
  by $e^2 \cdot n \exp(-n) + \exp(-n)$.
  Consequently, there exists a randomized simplex-type algorithm for
  linear programming with an expected  running time that is polynomial in
  $n/\delta$.
\end{theorem}

\subsection*{Remarks}

\begin{remark}
  The $\delta$-distance
  property is a global property of the matrix $A$.
  However, we used it only for the \emph{feasible bases} of the linear
  program when we showed rapid mixing of the random walk. In the first
  submission of our paper, we asked whether an analog of
  Lemma~\ref{lem:4} also holds for linear programs where a
  \emph{local} $\delta$-distance
  property holds. This was answered positively by Dadush and
  H\"ahnle~\cite{DBLP:conf/compgeom/DadushH15}. Thus our random walk
  can be used to solve linear programs whose basis matrices $A_B$,
  for each \emph{feasible basis} $B$,
  satisfy the $\delta$-distance
  property in expected polynomial time in $n/ \delta$.
  The paper of Dadush and H\"ahnle shows this as well via a variant of
  the \emph{shadow-vertex} pivoting rule. Their algorithm is much
  faster than ours.  Also Brunsch et. al~\cite{brunsch2015solving}
  could recently show that the shadow-vertex pivoting rule results in
  a polynomial-time algorithm for linear programs with the
  $\delta$-distance
  property. Their number of pivot-operations however also depends on
  the number $m$ of constraints.
\end{remark}

\begin{remark}
  Our random walk, and especially the partitioning scheme, is similar to Dyer and Frieze~\cite{MR1274170} and is directly inspired by their paper. The main  differences are the following. 
Dyer and Frieze rely on a \emph{finite} partitioning of a bounded subset of $\R^n$, adapting  a technique of Applegate and Kannan~\cite{ApplegateK91}. Our walk, however, is on a countably infinite set of states. For this, we rely on the paper of Lov\`asz and Simonovits~\cite{lovasz1993random} which extends the relation of the conductance and rapid mixing  to more general Markov chains, in particular countably infinite Markow chains. We also choose the $\ell_1$-norm, i.e., the function $\exp(-\|x - \alpha \cdot c\|_1)$,  whereas Dyer and Frieze used the Euclidean norm.  With this choice, the isoperimetric inequality of Bobkov and Houdr\'e~\cite{bobkov1997isoperimetric} makes the conductance analysis much simpler. Finally, we analyze the walk in terms of the $\delta$-distance in a satisfactory way. We feel that this geometric property is better suited for linear programming algorithms, as it is more general. 
\end{remark}

\subsubsection*{Acknowledgments} 

We are grateful to Daniel Dadush and Nicolai H\"ahnle, who
pointed out an error in the sub-division scheme in a previous version
of this paper. We also thank the anonymous referees for their careful remarks and suggestions that helped us to improve the presentation of our result. 

{ \small 

}

\section{Appendix}


\begin{proof}[Proof of Lemma~\ref{lem:4}]
  We denote the normal cones of $B$ and $B'$ by 
  \begin{displaymath}
    {C} = \{ \sum_{i\in B} \lambda_i a_i\colon \lambda_i\geq 0 \}  \quad \text{and} \quad  {C}' = \{\sum_{j \in B'} \mu_j a_j \colon \mu_j \geq 0\}. 
  \end{displaymath}
  By a gift-wrapping technique, we construct a hyperplane $(h^T x =
  0)$, $h \in \R^n \setminus \{0\}$  such that the following conditions hold. 
   \begin{enumerate}[i)]
  \item The hyperplane separates the interiors of $C$ and $C'$.\label{item:4}
  \item The row $a_k$ does not lie on the hyperplane. \label{item:6}
  \item The hyperplane is spanned by $n-1$ rows of A. \label{item:7}
  \end{enumerate}

Once we construct $(h^Tx = 0)$ we can argue as follows. 
The distance of $\mu_k \cdot a_k$ to the hyperplane $(h^Tx = 0)$ is at
least $\mu_k\cdot \delta$.  Since $c'$ is the sum of $\mu_k \cdot a_k$ and a
vector that is on the same side of the hyperplane as $a_k$ it
follows that the distance of $c'$ to this hyperplane is also at least
$\mu_k \cdot \delta$. Since $c$ lies on the opposite side of the
hyperplane, the distance of $c$ and $c'$ is at least $\mu_k \cdot
\delta$.

  We start with a hyperplane $(h^Tx = 0)$ strictly separating the
  interiors of $C$ and $C'$.  The conditions
  \ref{item:4},~\ref{item:6}) are satisfied. Suppose that \ref{item:7}) is not satisfied and let $\ell < n-1$ be the maximum number of linearly independent rows of $A$ that are contained in $(h^Tx = 0)$. 
  
  We tilt the hyperplane by moving its normal vector $h$ along a
  chosen equator of the ball of radius $\|h\|$ to augment this number.
  Since $\ell< n-1$ there exists  an equator leaving the rows
  of $A$ that are  in contained in $(h^T = 0)$ invariant under each rotation. 

  However, as soon as the hyperplane contains a new row of $A$ we
  stop.  If this new row of $A$ is not $a_k$ then still, conditions
  \ref{item:4},~\ref{item:6}) hold and the hyperplane now contains
  $\ell+1$ linearly independent rows of $A$. 
  
  If this new
  row is $a_k$, then we redo the tilting operation but this time by
  moving $h$ in the opposite direction on the chosen equator. Since
  there are $n$ linearly independent rows of $A$ without the row $a_k$
  this tilting will stop at a new row of $A$ which is not $a_k$ and we
  end the first tilting operation. 

  This tilting operation has to be repeated at most $n-1 - |B \cap B'|$ times to obtain the desired hyperplane.

\end{proof}

\subsection*{Phase 1}
\label{sec:finding-an-initial}

We now describe an approach to determine an initial basic feasible
solution or to assert that the linear program~(\ref{eq:1}) is
infeasible. Furthermore, we justify the assumption that the set of
feasible solutions is a bounded polytope. This \emph{phase~1} is
different from the usual textbook method since the linear programs
that we need to solve have to comply with the \emph{$\delta$-distance
  property}.

To find an initial basic feasible solution, we start by identifying
$n$ linearly independent linear inequalities $\wt{a}_1^Tx \leq \wt{b}_1,
\dots,\wt{a}_n^Tx \leq \wt{b}_n$ of $Ax \leq b$. Then we determine a ball that
contains all feasible solutions. This standard technique is for
example described in~\cite{GroetschelLovaszSchrijver88}. Using the
normal-vectors $\wt{a}_1,\dots,\wt{a}_n$ we next determine values $\beta_i,
\gamma_i \in \R$, $i=1,\dots,n$ such that this ball is contained in
the parallelepiped $Z = \{ x \in \R^n \colon \beta_i \leq \wt{a}_i^Tx \leq
\gamma_i, \, i=1,\dots,n\}$. We start with a basic feasible solution
$x_0^*$ of this parallelepiped and then proceed in $m$ iterations. In
iteration $i$, we determine a basic feasible solution $x^*_i$ of the
polytope
\begin{displaymath}
P_i =   Z \cap \{x \in \R^n \colon a_j^Tx \leq b_j, \,1 \leq j \leq i  \}
\end{displaymath}
using the basic feasible solution $x^*_{i-1}$ from the previous iteration by solving the linear program 
\begin{displaymath}
  \min\{a_i^Tx \colon x \in P_{i-1}\}. 
\end{displaymath}
If the optimum value of this linear program is larger than $b_i$, we
assert that the linear program~(\ref{eq:1}) is infeasible. Otherwise
$x^*_i$ is the basic feasible solution from this iteration.

Finally, we justify the assumption that $P = \{x \in \R^n \colon Ax
\leq b\}$ is bounded as follows. Instead of solving the linear
program~(\ref{eq:1}), we solve the linear program $\max\{ c^Tx \colon
x \in P \cap Z\}$ with the initial basic feasible solution $x^*_m$. If
the optimum solution is not feasible for (\ref{eq:1}) then we assert
that (\ref{eq:1}) is unbounded.

\end{document}